\documentclass{llncs}

\usepackage{amssymb,amsmath,stmaryrd}
\usepackage{latexsym}
\usepackage{amsfonts}
\usepackage{layout}
\usepackage{picinpar}
\usepackage[ruled, vlined]{algorithm2e}

\dontprintsemicolon
\usepackage{epic}
\usepackage{eepic}
\usepackage{graphicx}
\usepackage{comment}
\usepackage{gastex}
\usepackage{wrapfig}

\newcommand{\abs}[1]{\ensuremath{\lvert #1\rvert}}
\newcommand{\tuple}[1]{\langle #1 \rangle}
\newcommand{\C}{{\sf C}}

\newcommand{\nat}{\mathbb N} 
\newcommand{\zed}{{\mathbb Z}}

\newcommand{\set}[1]{\{#1\}}

\newcommand{\rat}{{\mathbb Q}}

\newcommand{\real}{{\mathbb R}}

\newcommand{\straa}{\sigma}

\newcommand{\strab}{\pi}

\newcommand{\Attr}{Attr}

\newcommand{\Plays}{{\sf Plays}}
\newcommand{\Inf}{{\sf Inf}}
\newcommand{\EL}{{\sf EL}}
\newcommand{\MP}{{\sf MP}}
\newcommand{\Parity}{{\sf Parity}}
\newcommand{\PosEnergy}{{\sf PosEnergy}}
\newcommand{\MeanPayoff}{{\sf MeanPayoff}}
\newcommand{\M}{{\sf M}}
\newcommand{\Bad}{{\sf Bad}}
\newcommand{\Win}{{\sf Win}}
\newcommand{\gfe}{\hspace{-1pt}{\it gfe}}
\newcommand{\GFE}{{\sf GFE}}

\title{Energy Parity Games\thanks{This is an improved version of a paper that appeared in the 
\emph{Proceedings of the 37th International Colloquium on Automata, Languages and Programming} (ICALP),
Lecture Notes in Computer Science 6199, Springer-Verlag, 2010, pages 599-610. The present version
contains detailed proofs, and improved memory and algorithmic complexity bounds.
}}

\author{Krishnendu Chatterjee\inst{1} \and Laurent Doyen\inst{2}}

\institute{
IST Austria (Institute of Science and Technology Austria) \\
\and LSV, ENS Cachan \& CNRS, France 
}

\begin{document}
\maketitle

\begin{abstract} 
Energy parity games are infinite two-player turn-based games played on weighted graphs.
The objective of the game combines a (qualitative) parity condition with 
the (quantitative) requirement that the sum of the weights (i.e., the
level of energy in the game) must remain positive. 
Beside their own interest in the design and synthesis of resource-constrained 
omega-regular specifications, energy parity games provide one of the simplest 
model of games with combined qualitative and quantitative objective. 
Our main results are as follows: (a) exponential memory is sufficient and 
may be necessary for winning strategies in energy parity games; 
(b) the problem of deciding the winner in energy parity games can be solved 
in NP~$\cap$~coNP; and 
(c) we give an algorithm to solve energy parity by reduction to energy games. 
We also show that the problem of deciding the winner in energy parity games is 
polynomially equivalent to the problem of deciding the winner in mean-payoff
parity games, which can thus be solved in NP~$\cap$~coNP.
As a consequence we also 
obtain a conceptually simple algorithm to solve 
mean-payoff parity games.

\end{abstract}


\section{Introduction}
Two-player games on graphs are central in many applications
of computer science.  For example, in the synthesis problem implementations are 
obtained from winning strategies in games with a qualitative objective such as 
$\omega$-regular specifications~\cite{RW87,PnueliR89,AbadiLW89}.
Games also provide a theoretical instrument to deal with logics and 
automata~\cite{BL69,GurevichHarrington82,EmersonJS93,automata}.
In all these applications, the games have a qualitative (boolean) objective 
that determines which player wins. 
On the other hand, games with quantitative objective which are natural models 
in economics (where players have to optimize a real-valued payoff) have also been studied 
in the context of automated design and synthesis~\cite{Sha53,Condon92,ZwickP96}.
In the recent past, there has been considerable interest in the design of 
reactive systems that work in resource-constrained environments 
(such as embedded systems).
The specifications for such reactive systems have both a quantitative 
component (specifying the resource constraint such as limited power consumption) and 
a qualitative component (specifying the functional requirement).
The desired reactive system must respect both the qualitative and quantitative
specifications.
Only recently objectives combining both qualitative and quantitative 
specifications have been considered~\cite{CAHS03,ChatterjeeHJ05,BCHJ09}.

In this paper, we consider two-player turn-based games played for infinitely 
many rounds on a weighted graph where a \emph{priority} is associated to each 
state and an integer \emph{weight} (encoded in binary) is associated to each edge.
In each round, the player owning the current state chooses an outgoing edge
to a successor state, thus the game results in an infinite play.
The qualitative specification is a \emph{parity} condition, 
a canonical way to express the $\omega$-regular objectives~\cite{Thomas97}.
A play satisfies the parity condition if the least priority occurring 
infinitely often in the play is even. The quantitative specification is an \emph{energy} 
condition which requires 
that the sum of the weights along the play (that we interpret as the level of 
energy, or resource usage) remains always positive. 
Energy parity games can be viewed as games played on one-counter automata 
with fairness condition. 
The main algorithmic question about energy parity games
is to decide if there exists an initial credit (or initial energy level) such 
that one player has a strategy to maintain the level of energy positive while
satisfying the parity condition, and if the answer is yes, to compute the 
minimum such initial credit.

Energy parity games generalize both parity games and energy games. It is known
that memoryless strategies are sufficient to win parity games~\cite{EJ91}
and energy games~\cite{CAHS03,BFLMS08}, and therefore the problem of deciding 
the winner of a parity game, and the problem of deciding the existence of an 
initial credit sufficient to win an energy game are both in NP~$\cap$~coNP. 
It is a long standing open question to know if these problems can be solved in 
polynomial time.
In this paper, we present the following results about energy parity games: 
(1) we study the complexity of winning strategies and we give bounds on the 
amount of memory needed to win;
(2) we establish the computational complexity of the problem of deciding the winner; 
(3) we present an algorithmic solution to compute the minimum initial credit; and 
(4) we show polynomial equivalence with mean-payoff parity games. 
The details of our contributions are as follows. 
\begin{enumerate}

\item \emph{Strategy complexity.}
First, we show that finite-memory strategies are sufficient to win energy 
parity games, but exponential memory may be required even in the special 
case of one-player games.
We present an exponential memory upper bound for the winning strategies.
Our memory bound is $n\cdot d \cdot W$, where $n$ is the size of the state
space, $d$ is the number of priorities, and 
$W$ is the maximum absolute value of the weights.
This bound is exponential since $W$ can be encoded in $\log W$ bits,
but polynomial in $n$ and $d$.
We show that the spoiling strategies of the opponent need no 
memory at all (memoryless spoiling strategies exist). 

\item \emph{Computational complexity.} 
Second, we show that the decision problem for energy parity games lie in NP $\cap$ coNP,
matching the bounds known for the simpler case of parity and energy games. 
The classical NP~$\cap$~coNP result for parity and energy games crucially relies
on the existence of memoryless winning strategies. 
In the case of energy parity games, the existence of memoryless spoiling strategies
gives the coNP upper bound.
However, and in contrast with parity games and energy games, 
winning strategies may require exponential memory in energy parity games. 
Therefore, more subtle arguments are needed to obtain the NP upper bound: we show that the 
winning strategies (that require exponential memory) can be characterized with 
certain special structures and decomposed into two memoryless strategies 
(roughly, one to ensure the parity condition, and the other to maintain the 
energy level positive). 
This insight allows us to derive a nondeterministic polynomial-time
algorithm to solve energy parity games. 
Thus the problem of deciding the existence of an initial credit which is 
sufficient to win an energy parity game is (perhaps surprisingly) in 
NP~$\cap$~coNP. 
Finding a deterministic polynomial algorithm for this problem is obviously open.

\item \emph{Algorithm.}
Third, we present an algorithm to solve energy parity games with complexity
exponential in the number of states~$n$ (as for parity games), and linear in the largest
weight~$W$ (as for energy games). This algorithm relies on our analysis of the
structure of winning strategies, and reduces to iteratively solving reachability games
and energy games.

\item \emph{Equivalence with mean-payoff parity games.} 
Finally, we show that energy parity games are polynomially equivalent 
to mean-payoff  parity games~\cite{ChatterjeeHJ05}, where the parity condition is combined 
with the quantitative requirement that the limit-average (or mean-payoff) of the weights 
remains positive. 
Again, this result is surprising because in mean-payoff parity games, optimal strategies 
(that realize the largest possible mean-payoff value while satisfying the parity condition) may require 
infinite memory. Moreover, we get as a corollary of our results that the problem of deciding 
the winner in mean-payoff parity games is also in NP~$\cap$~coNP. 
Our algorithm for energy parity games can be used to solve mean-payoff parity games
with essentially the same complexity as in~\cite{ChatterjeeHJ05}, 
but with a conceptually simpler approach. 
\end{enumerate}

\smallskip\noindent\emph{Relation to one-counter parity games.} 
Energy parity games can be reduced to one-counter parity games~\cite{Serre06},
where the counter can be incremented and decremented only by~1 (i.e., the 
weights are in $\set{-1,0,1}$). 
Since the weights in energy parity games are encoded succinctly in binary, the 
reduction is exponential. 
It was shown that exponential memory is sufficient in one-counter parity games, 
and that they can be solved in PSPACE and are DP-hard~\cite{Serre06},
showing that  one-counter parity games are more general.
The exponential reduction and results on one-counter parity games
would give an EXPSPACE upper bound for the problem, and a double exponential ($2^{n \cdot d \cdot W}$) 
upper bound on memory, whereas we show that the problem is 
in NP $\cap$ coNP, and exponential memory is sufficient.

\section{Definitions}

\noindent{\bf Game graphs.}
A \emph{game graph} $G=\tuple{Q, E}$ consists of a finite set $Q$ of states 
partitioned into \mbox{player-$1$} states $Q_1$ and player-2 states $Q_2$ (i.e., $Q=Q_1 \cup Q_2$), 
and a set $E \subseteq Q \times Q$ of edges such that for all $q \in Q$,
there exists (at least one) $q' \in Q$ such that $(q,q') \in E$.
A \emph{player-$1$ game} is a game graph where $Q_1 = Q$ and $Q_2 = \emptyset$.
The subgraph of $G$ induced by $S \subseteq Q$ is the graph 
$\tuple{S, E \cap (S \times S)}$ (which is not a game graph in 
general); the subgraph induced by $S$ is a game graph if for all 
$s \in S$ there exist $s' \in S$ such that $(s,s')\in E$.

\smallskip\noindent{\bf Plays and strategies.}
A game on $G$ starting from a state $q_0 \in Q$ is played in rounds as follows. 
If the game is in a player-1 state, then player~$1$ chooses the successor state from the set
of outgoing edges; otherwise the game is in a player-$2$ state, and player $2$ chooses the successor 
state. 
The game results in a \emph{play} from~$q_0$, i.e., 
an infinite path $\rho = q_0 q_1 \dots$ such that $(q_i,q_{i+1}) \in E$ for all $i \geq 0$. 
The prefix of length $n$ of $\rho$ is denoted by $\rho(n) = q_0 \dots q_n$. 
The \emph{cycle decomposition} of $\rho$ is an infinite sequence of simple cycles $C_1, C_2, \dots$
obtained as follows: push successively $q_0, q_1, \dots$ onto a stack, and 
whenever we push a state already in the stack, a simple cycle is formed that we remove 
from the stack and append to the cycle decomposition. Note that the stack content
is always a prefix of a path of length at most~$\abs{Q}$.

A \emph{strategy} for player~$1$ is a function
$\straa: Q^*Q_1 \to Q$ such that $(q,\straa(\rho\cdot q)) \in E$ for all $q \in Q_1$
and all $\rho \in Q^*$. An \emph{outcome} of $\straa$ from~$q_0$ is a play $q_0 q_1 \dots$ such that 
$\straa(q_0 \dots q_i) = q_{i+1}$ for all $i \geq 0$ such that $q_i \in Q_1$. Strategy and outcome for
player~$2$ are defined analogously.

\smallskip\noindent{\bf Finite-memory strategies.}
A strategy uses \emph{finite-memory} if it can be encoded
by a deterministic transducer $\tuple{M, m_0, \alpha_u, \alpha_n}$
where $M$ is a finite set (the memory of the strategy), $m_0 \in M$ is the initial memory value,
$\alpha_u: M \times Q \to M$ is an update function, and $\alpha_n: M \times Q_1 \to Q$ is a next-move function. 
The \emph{size} of the strategy is the number $\abs{M}$ of memory values.
If the game is in a player-$1$ state $q$ and $m$ is the current memory value,
then the strategy chooses $q' = \alpha_n(m,q)$ as the next 
state and the memory is updated to $\alpha_u(m,q)$. 
Formally, $\tuple{M, m_0, \alpha_u, \alpha_n}$
defines the strategy $\alpha$ such that $\alpha(\rho\cdot q) = \alpha_n(\hat{\alpha}_u(m_0, \rho), q)$
for all $\rho \in Q^*$ and $q \in Q_1$, where $\hat{\alpha}_u$ extends $\alpha_u$ to sequences
of states as expected. A strategy is \emph{memoryless} if $\abs{M} = 1$.
For a finite-memory strategy $\straa$, let $G_{\straa}$ be the graph obtained as the product
of $G$ with the transducer defining $\straa$, where $(\tuple{m,q},\tuple{m',q'})$ is a transition
in $G_{\straa}$ if $m' = \alpha_u(m,q)$ and either $q \in Q_1$ and $q'=\alpha_n(m,q)$, or $q \in Q_2$ and $(q,q') \in E$.
In $G_{\straa}$, the expression \emph{reachable from $q$} stands for \emph{reachable from $\tuple{q,m_0}$}.

\smallskip\noindent{\bf Objectives.}
An \emph{objective} for $G$ is a set $\phi \subseteq Q^\omega$. 
Let $p:Q \to \nat$ be a \emph{priority function} and $w:E \to \zed$ be a \emph{weight function}\footnote{In some proofs, 
we take the freedom to use rational weights (i.e., $w:E \to \rat$), while we always assume that weights are integers encoded in binary for complexity results.}  
where positive numbers represent rewards, and negative numbers represent costs. 
We denote by $W$ the largest weight of an edge (in absolute value) according to $w$.
The \emph{energy level} of a prefix $\gamma = q_0 q_1 \dots q_n$ of a play
is $\EL(w,\gamma) = \sum_{i=0}^{n-1} w(q_i,q_{i+1})$, and the \emph{mean-payoff value}
of a play $\rho= q_0 q_1 \dots$ is 
$\MP(w,\rho) = \liminf_{n \to \infty} \frac{1}{n}\cdot\EL(w,\rho(n))$. 
In the sequel, when the weight function~$w$ is clear from context we will omit it and
simply write $\EL(\gamma)$ and $\MP(\rho)$.
We denote by $\Inf(\rho)$ the set of states that occur infinitely often in $\rho$.
We consider the following objectives:

\begin{itemize}
 	\item \emph{Parity objectives.}
	The \emph{parity} objective $\Parity_G(p) = \{\rho \in \Plays(G) \mid \min\{p(q) \mid q \in \Inf(\rho)\} \text{ is even }\}$
	requires that the minimum priority visited infinitely often be even. 
	The special cases of \emph{B\"uchi} and \emph{coB\"uchi} objectives correspond
	to the case with two priorities, $p: Q \to \{0,1\}$ and $p: Q \to \{1,2\}$ respectively.

	\item \emph{Energy objectives.}
	Given an initial credit $c_0 \in \nat \cup \{\infty\}$, the \emph{energy} objective 
	$\PosEnergy_G(c_0) = \{ \rho \in \Plays(G) \mid \forall n \geq 0:  c_0 + \EL(\rho(n)) \geq 0 \}$
	requires that the energy level be always positive.

	\item \emph{Mean-payoff objectives.}
	Given a threshold $\nu \in \rat$, the \emph{mean-payoff} objective 
	$\MeanPayoff_G(\nu) = \{ \rho \in {\sf Plays}(G) \mid \MP(\rho) \geq \nu \}$
	requires that the mean-payoff value be at least $\nu$.

	\item \emph{Combined objectives.}
	The \emph{energy parity} objective $\Parity_G(p) \cap \PosEnergy_G(c_0)$
	and the \emph{mean-payoff parity} objective $\Parity_G(p) \cap \MeanPayoff_G(\nu)$
	combine the requirements of parity and energy (resp., mean-payoff) objectives.
\end{itemize}
\noindent
Note that parity objecitves are \emph{prefix-independent}, i.e. for all plays
$\rho$ and $\rho'$ such that $\rho' = \gamma \cdot \rho$ where $\gamma$ is a finite prefix,
we have $\rho \in \Parity_G(p)$ iff $\rho' \in \Parity_G(p)$.

When the game~$G$ is clear form the context, we omit the subscript in objective names.

\smallskip\noindent{\bf Winning strategies.}
A player-$1$ strategy~$\straa$ is \emph{winning}\footnote{We also say that player-$1$ is winning, or that $q$ is a winning state.} 
in a state~$q$ for an objective~$\phi$ if 
$\rho \in \phi$ for all outcomes~$\rho$ of~$\straa$ from~$q$.
For energy and energy parity objectives with unspecified initial credit, we also say that a strategy is winning
if it is winning for some finite initial credit.

\smallskip\noindent{\bf Finite and minimum initial credit problems.}
We are interested in the following decision problem.
The \emph{finite initial credit problem} (initial credit problem for short) 
asks, given an energy parity game $\tuple{G,p,w}$ and a state~$q$, 
whether there exists a finite initial credit $c_0 \in \nat$ and a 
winning strategy for player~$1$ from~$q$  with initial credit~$c_0$.
The \emph{minimum initial credit} in a state $q \in Q$
is the least value of initial credit for which there exists a winning strategy 
for player~$1$ in $q$.
A strategy for player~$1$ is \emph{optimal} in a state $q$ if it is winning 
from $q$  with the minimum initial credit.

It is known that the initial credit problem for simple energy games can be solved in NP~$\cap$~coNP
because memoryless strategies are sufficient to win such games~\cite{CAHS03,BFLMS08}.
For winning states of energy games, an initial credit of $(\abs{Q} -1)\cdot W$ is always sufficient 
to win.
For parity games, memoryless strategies are also sufficient to win and the associated decision
problem also lies in NP~$\cap$~coNP~\cite{EJ91}. 
Moreover, energy games and parity
games are determined, which implies that from states that are not winning for player~$1$,
there exists a (memoryless) spoiling strategy for player~$2$ which is winning for the complementary objective
(note that the complement of a parity objective is again a parity objective). Moreover, for energy
games, the same spoiling strategy can be used against all initial credit values.

\section{Strategy Complexity of Energy Parity Games}\label{sec:memory-credit}

In this section we show that in energy parity games with $n$ states and $d$ priorities,
memory of size $n \cdot d \cdot W$ is sufficient for a winning strategy of player~$1$.
This amount of memory is exponential (because weights are encoded in binary) and
we show that exponential memory is already necessary in the special case of player-1 games
with two priorities where memory of size $2 \cdot (n-1) \cdot W + 1$ may be necessary (and is always sufficient).
For player~$2$, we show that memoryless winning strategies exist.
Moreover, if player~$1$ wins, then the minimum initial credit is always at most $(n-1) \cdot W$.

\begin{figure}[!tb]
  \begin{center}
    \hrule
    \begin{picture}(105,25)(0,0)


\node[Nmarks=ir,NLangle=0.0](n0)(10,10){$0$}
\node[Nmarks=n](n1)(30,10){$1$}

\node[Nframe=n](phantom1)(50,10){}
\node[Nframe=n](dots)(55,10){$\dots$}
\node[Nframe=n](phantom2)(60,10){}

\node[Nmarks=n](n3)(80,10){$1$}
\node[Nmarks=n](n4)(100,10){$1$}

\drawloop[ELside=l,loopCW=y, loopdiam=6](n4){$1$}


\drawedge[ELpos=50, ELside=l, curvedepth=4](n0,n1){$-W$}
\drawedge[ELpos=50, ELside=l, curvedepth=4](n1,n0){$-W$}

\drawedge[ELpos=50, ELside=l, curvedepth=4](n1,phantom1){$-W$}
\drawedge[ELpos=50, ELside=l, curvedepth=4](phantom1,n1){$-W$}

\drawedge[ELpos=50, ELside=l, curvedepth=4](phantom2,n3){$-W$}
\drawedge[ELpos=50, ELside=l, curvedepth=4](n3,phantom2){$-W$}

\drawedge[ELpos=50, ELside=l, curvedepth=4](n3,n4){$-W$}
\drawedge[ELpos=50, ELside=l, curvedepth=4](n4,n3){$-W$}


\end{picture}
    \hrule
    \caption{A family of $1$-player energy parity games where Player~$1$ needs 
	memory of size $2 \cdot (n-1) \cdot W$ and initial credit 
	$(n-1) \cdot W$.
	Edges are labeled by weights, states by priorities.  \label{fig:initial-credit}}
  \end{center}
\end{figure}
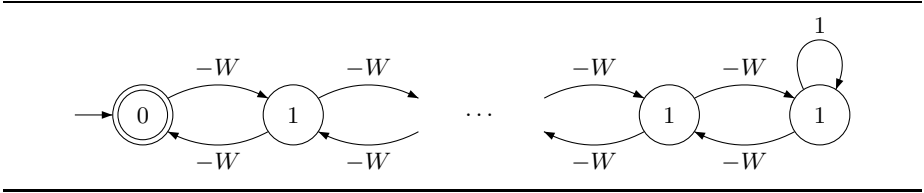


\begin{lemma}\label{lem:one-player-epg}
Let $G$ be a player-$1$ energy parity game with $n$ states. 
If player~$1$ wins in $G$ from a state $q_0$, then player~$1$ has a winning strategy 
from $q_0$ with memory of size $2 \cdot (n-1) \cdot W + 1$ and initial credit
$(n-1) \cdot W$. 
\end{lemma}

\begin{proof}
Since $G$ is a player-$1$ energy parity game, we have $Q_1 = Q$ and $Q_2 = \emptyset$.
Consider an outcome~$\rho$ of an optimal strategy for player~$1$ in~$G$.
Note that the minimal priority of the states in $\Inf(\rho)$ is even, 
that $\Inf(\rho)$ is strongly connected, and that there exists a suffix 
$\rho'$ of $\rho$ that only contains states in $\Inf(\rho)$.
Let $C_1, C_2, \dots$ be the cycle decomposition of $\rho'$. We consider two cases. 

First, if $\EL(C_i) > 0$ for some cycle $C_i$,
then we construct a winning strategy for player~$1$ as follows. From the starting state,
reach a state of $C_i$ and go through $C_i$ once. This can be done with initial credit 
$(n-1) \cdot W$.
Now, pump the cycle to get the energy level above $2 \cdot (n-1) \cdot W$,
and then reach a state of $\Inf(\rho)$ with minimal priority (this consumes at 
most $(n-1) \cdot W$ units of energy) and go back to the cycle (which also 
consumes at most $(n-1) \cdot W$ units of energy). Hence, at this point 
the energy level is still positive, and we can iterate $(i)$ pumping the positive
cycle, $(ii)$ reach the minimal even priority, and $(iii)$ go back to the cycle.
This defines a winning strategy with memory of size $2 \cdot (n-1) \cdot W + 1$ and initial credit
$(n-1) \cdot W$. 

Second, if $\EL(C_i) \leq 0$ for all cycles $C_i$ ($i \geq 1$),
then it is easy to see that there exists $k \geq 1$ such that 
$\EL(C_j) = 0$ for all $j \geq k$. Since the parity condition 
is satisfied in $\rho$, the minimal priority of the states in $\Inf(\rho)$
is visited by some cycle $C_j$ ($j \geq k$). We construct a winning strategy 
for player~$1$ as follows. From the starting state, reach a state of $C_j$ 
and go through $C_j$ forever. This can be done with initial credit $(n-1) \cdot W$
and it is clearly a winning strategy.

In both cases, player~$1$ wins with memory of size $2 \cdot (n-1) \cdot W + 1$ and initial credit
$(n-1) \cdot W$.
\qed
\end{proof}

\begin{example}[Memory requirement]\label{examp:memory}
We present a family of player-$1$ games where memory of size  $2 \cdot (n-1) \cdot W + 1$ 
may be necessary. 
The example is shown in \figurename~\ref{fig:initial-credit}, and the example 
also shows that initial credit of $(n-1) \cdot W$ may be necessary.
To satisfy the parity condition, the play has to visit the initial state infinitely often, 
and to maintain the energy positive, the play has to visit the state with the positive-weighted self-loop.
Since the paths between these two state have weight $-(n-1) \cdot W$, it is easy to see that
initial credit $(n-1) \cdot W$ is necessary, and the self-loop has to be taken $M=2 \cdot (n-1) \cdot W$ times
requiring memory of size $M+1$.
\qed
\end{example}

We state the next lemma because it is useful in several proofs, though its argument is fairly easy.

\begin{lemma}\label{lem:energy-is-monotone}
Let $G$ be an energy parity game, and for each winning state $q$ let $v(q) \in \nat$ be the
minimum initial credit in $q$. 
For all outcomes $\rho$ of an optimal strategy $\straa$ in $G$ from a winning state $q_0$,
if the initial credit is $v(q_0) + \Delta$ for $\Delta \geq 0$, 
then the energy level at all positions of $\rho$ where a state $q$ occurs is at least $v(q) + \Delta$.
\end{lemma}

\begin{proof}
It is easy to see that for all outcomes $\rho$ of $\straa$ in $G$, the energy 
level at all positions of $\rho$ where $q$ occurs must be at least $v(q)$ (otherwise
if a $q$-position has energy level below $v(q)$, then player~$2$ can win from that
position, and therefore wins in the original game in contradiction with optimality of $\straa$). 
Hence, since strategies are functions of sequence of states only (and not of their energy 
level), if we start with energy level $v(q_0) + \Delta$, then the energy 
level at all positions of an outcome of $\straa$ is greater by $\Delta$ than if we had started 
with energy level $v(q_0)$. In particular, for all positions where $q$ occurs in
an outcome of~$\straa$, the energy level is at least $v(q) + \Delta$.
\qed
\end{proof}

We show that player~$2$ needs no memory at all in energy parity games. 
Note that energy objectives are not prefix-independent objectives
and the following lemma does not directly follow from the results 
of~\cite{Kop06}. 
However our proof, which is based on induction on edges, is an adaptation 
of the proof technique of~\cite{Kop06,GimbertZ05}. 
This result is useful to show that energy parity games are in coNP.

\begin{lemma}\label{lem:player-two-memoryless}
For all energy parity games $G$, memoryless strategies are sufficient for 
player~$2$ (i.e., the minimum initial credit for player~$1$ does not change
if player~$2$ is restricted to play memoryless).
\end{lemma}

\begin{proof}
Without loss of generality, we assume that every player-$2$ state has two outgoing
edges. The proof is by induction on the number of player-$2$ states. If $\abs{Q_2} = 0$,
then the result is trivial. Assume that the result holds for all energy parity games 
with $\abs{Q_2} < k$ and let $G$ be an energy parity games with $\abs{Q_2} = k$.

Consider some player-$2$ state $\hat{q}$ with outgoing edges 
$e_l = (\hat{q},q_l)$ and $e_r = (\hat{q},q_r)$. 
Let $G_l$ and $G_r$ be the game graphs obtained from~$G$ by
removing the edges $e_r$ and $e_l$ respectively. By the induction hypothesis,
memoryless strategies are sufficient for player~$2$ in~$G_l$ and~$G_r$.
For each $q \in Q$, let $v_l(q)$ and $v_r(q)$ be the minimal initial credit 
for player~$1$ from~$q$ in $G_l$ and $G_r$ respectively, and let $\straa_{l}$
and $\straa_{r}$ be corresponding optimal strategies for player~$1$. 
Assume without loss of generality that $v_l(\hat{q}) \geq v_r(\hat{q})$.

First, we show that for all $q \in Q$ the initial credit $v_l(q)$ in $q$ is sufficient 
to win in $G_r$, i.e., $v_l(q) \geq v_r(q)$  $(\star)$. 
To obtain this, we play in $G_r$ from~$q$ as would play an optimal strategy in $G_l$ and 
if we reach $\hat{q}$, then we play an optimal strategy starting from $\hat{q}$ in $G_r$. 
Consider an outcome~$\rho \in Q^\omega$ of this strategy. 
Either $\rho$~never visits $\hat{q}$ and then the initial credit $v_l(q)$ is clearly 
sufficient to win, or $\rho$~eventually visits $\hat{q}$ once and then the energy 
level is at least $v_l(\hat{q}) \geq v_r(\hat{q})$ by Lemma~\ref{lem:energy-is-monotone} 
in $G_l$ (since we played as in $G_l$ so far). 
Since from there on we play as in $G_r$, the energy level of $\rho$ never drops
below $0$, and the parity condition (in the whole play) is satisfied since 
it is satisifed in a suffix, and parity is a prefix-independent objective.

Second, we construct a strategy $\straa_{lr}$ for player~$1$ in $G$ that wins with initial credit
$\max\{v_l(q), v_r(q)\}$ from every $q \in Q$, establishing the result. 
Given a prefix $\tau \in Q^*Q_1$,
if $\hat{q}$ does not occur in $\tau$, then the strategy plays as in $G_l$, i.e., 
$\straa_{lr}(\tau) = \straa_{l}(\tau)$.
If $\hat{q}$ occurs in $\tau$, then we decompose $\tau$ into segments as follows: 
a finite prefix before the first visit to $\hat{q}$, 
then a (possibly empty) sequence of cycles over $\hat{q}$ (these cycles are not 
necessarily simple, but they do not contain nested cycles over $\hat{q}$), 
and then a (possibly empty) finite suffix after the last visit to $\hat{q}$. 
We label the cycles and the suffix with $l$ if $e_l$ was taken from $\hat{q}$, 
and with $r$ if $e_r$ was taken. 
If the last segment in $\tau$ is labeled by $d \in \{l,r\}$,
then the strategy for player~$1$ in~$G$ plays as the optimal strategy in $G_d$ applied
to the prefix~$\tau_d$ obtained from~$\tau$ by taking out the finite prefix and 
all segments not labeled by~$d$, i.e. $\straa_{lr}(\tau) = \straa_{d}(\tau_d)$. 

Now for all $q \in Q$, we show that $\straa_{lr}$ is winning in $G$ from $q$ with
initial credit $v(q) = \max\{v_l(q), v_r(q)\}$, i.e., we show that $v(q) = v_l(q)$ (by $(\star)$). 
Note that if $v_l(q) = \infty$ (or $v_r(q) = \infty$), 
then clearly $v(q) = \infty$ against player~$2$ playing as in $G_l$ (or $G_r$). 
So, we assume that $v_l(q)$ and $v_r(q)$ are finite.
Let $\rho$ be an outcome of the strategy $\straa_{lr}$. 
If~$\rho$ never visits~$\hat{q}$, then $\straa_{lr}$ has played as $\straa_{l}$ 
and the initial credit~$v_l(q)$ is sufficient to win. 
If~$\rho$ visits $\hat{q}$, then we decompose $\rho$ into segments as above 
(there may be no ``suffix'' if $\hat{q}$ is visited infinitely often) and we obtain
$\rho_d$ for $d \in \{l,r\}$ by removing from $\rho$ the prefix up to the first visit
to $\hat{q}$, and all segments not labeled by~$d$.
Note that the initial state is $\hat{q}$ in both $\rho_l$ and $\rho_r$.
Since the initial credit in $q$ is $v_l(q)$, we know that the energy level in
the first visit to $\hat{q}$ in $\rho$ is at least $v_l(\hat{q}) \geq v_r(\hat{q})$
(since $\straa_{lr}$ played as $\straa_{l}$ in $G_l$ so far).
By definition of $\straa_{lr}$, we also know that $\rho_l$ and $\rho_r$ are outcomes 
of optimal strategies in $G_l$ and $G_r$ respectively. 
Therefore the energy level in every position
of $\rho_l$ and $\rho_r$ where state $\hat{q}$ occurs is greater than the 
energy level in their initial position (using Lemma~\ref{lem:energy-is-monotone}). 
We say that the \emph{effect} of $\rho_l$ and $\rho_r$ on the energy level 
in $\hat{q}$ is nonnegative.

Therefore, if we consider the positions in $\rho$ where $\hat{q}$ occurs, 
if the position is in a $d$-labeled segment ($d \in \{l,r\}$), then the energy level 
is at least the energy level in the corresponding position in $\rho_d$ 
(because the effect on the energy level of the $\bar{d}$-labeled segments before 
that position is nonnegative - where $\bar{d} = l$ if $d = r$ and vice versa).
Therefore, the energy level in $\rho$ never drops below $0$. Moreover,
among $\rho_l$ and $\rho_r$, those that are infinite satisfy the parity condition,
so that that $\rho$ also satisfies the parity condition. Hence, $\rho$ satisfies
the energy parity condition.
\qed
\end{proof}

Finally, we give upper bounds on the memory and initial credit necessary for player~$1$
in energy parity games. The bounds are established using strategies of a special
form that alternate between \emph{good-for-energy} strategies and \emph{attractor}
strategies, defined as follows.

\smallskip\noindent{\bf Good-for-energy strategy.}
A strategy $\straa$ for player~$1$ is \emph{good-for-energy} in state $q$ if for all outcomes 
$\rho = q_0 q_1 \dots$ of $\straa$ such that $q_0 = q$,
for all cycles $C$ in the cycle decomposition of $\rho$, 
either $\EL(C) > 0$,
or $\EL(C) = 0$ and $C$ is even (i.e., $\min \{p(q) \mid q \in C \}$ is even).
A key result is to show the existence of good-for-energy strategies that are \emph{memoryless}.

\begin{lemma}\label{lem:good-for-energy}
Let $\Win$ be the set of winning states for player~1 in an energy parity game.
Then, there exists a memoryless strategy for player~1 which is good-for-energy
in every state $q \in ‌\Win$.
\end{lemma}

\begin{proof}
First, the definition of good-for-energy strategy in a state $q$ can be viewed as a winning strategy
in a finite cycle-forming game from $q$ where the game stops when a cycle $C$ is formed,
and the winner is determined by the sequence of states in $C$ (and is independent of cyclic
permutations). By the results of~\cite{BjorklundSV04}, both players have memoryless optimal
strategies in this finite cycle-forming game. 

Now, assume that player~$1$ wins an energy parity game from a state $q$. Towards contradiction, 
assume that player~$1$ has no good-for-energy strategy from $q$. Then, player~$2$ would have a 
memoryless winning strategy in the finite cycle-forming game. Fix this strategy in the original
energy parity game and then all cycles have either negative weight, or weight is zero and the 
least priority is odd. It follows that player~$1$ looses the energy parity game from $q$
(no matter the value of the initial credit), a contradiction. Hence, player~$1$ has a memoryless
good-for-energy strategy~$\sigma_q$ from~$q$.
Finally, to obtain a uniform good-for-energy strategy $\sigma_{\gfe}$, fix a (total) order on the states: $q_1 < q_2 < \cdots < q_n$, 
and let $R(q_i)$ be the set of all states occurring in the outcomes of $\sigma_{q_i}$.
Then $\sigma_{\gfe}(q_i) = \sigma_{q_j}(q_i)$ where $j = \min \{k \mid q_i \in R(q_k)\}$.
\qed
\end{proof}

\noindent{\bf Attractor.}
The player-$1$ \emph{attractor} of a given set $S \subseteq Q$ is the set of
states from which player~$1$ can force to reach a state in $S$. It is defined
as the limit $\Attr_1(S)$ of the sequence $A_0 = S$, $A_{i+1} = A_i \cup 
\{q \in Q_1 \mid \exists (q,q') \in E: q' \in A_i \} \cup 
\{q \in Q_2 \mid \forall (q,q') \in E: q' \in A_i \}$ for all $i \geq 0$.
The player-$2$ \emph{attractor} $\Attr_2(S)$ is defined symmetrically.
Note that for $i=1,2$, the subgraph of $G$ induced by $Q \setminus \Attr_i(S)$ 
is again a game graph (i.e., every state has an outgoing edge). It is well known
that attractors can be computed in polynomial time.

\begin{lemma}\label{lem:two-player-epg}
For all energy parity games $G$ with $n$ states and $d$ priorities,
if player~$1$ wins from a state $q_0$, then player~$1$ has a winning strategy from $q_0$ with memory of size 
$n \cdot d \cdot W$ and initial credit $(n-1) \cdot W$.
\end{lemma}

\begin{proof}[of Lemma~\ref{lem:two-player-epg}]
We prove by induction a slightly stronger statement, namely that player~$1$ has 
a winning strategy with memory of size $n \cdot d \cdot W$, where
$n=\abs{Q}$, and such that all its outcomes with initial credit 
$x \geq (n-1) \cdot W$
have energy level always at least $x - (n-1) \cdot W$ 
(and this strategy is winning from every state where player~$1$ wins in $G$, 
thus including $q_0$).
 
For the case of $d=1$ priority, either the priority is odd and all states are loosing for player~$1$ (hence, the result holds trivially),
or the priority is even and the energy parity game reduces to an energy game which can be won by player~$1$ with a memoryless
strategy and initial credit $(n-1) \cdot W$ from every winning state 
(\cite{CAHS03,BFLMS08,DGR09}). 
By Lemma~\ref{lem:energy-is-monotone}, if the initial credit is $x \geq (n-1) \cdot W$,
then the same strategy ensures that the energy level is always at least 
$x - (n-1) \cdot W$.

By induction, assume that the statement holds for all energy parity games~$G$ with $d-1$ priorities.
Consider a winning state~$q_0$ in an energy parity game~$G$ with $d$ priorities. 
By Lemma~\ref{lem:good-for-energy}, player~$1$ has a memoryless strategy $\straa_{\gfe}$
which is good-for-energy from every winning state of $G$. We consider two cases.

\noindent{{\bf A.}} If the least priority in $G$ is even (say it is $0$). 
Let $\Win$ be the set of winning states for player~$1$ in $G$ (thus $q_0 \in \Win$),
and let $\Omega_0$ be the player-$1$ attractor of priority-$0$ states in the subgraph of $G$ induced 
by $\Win$. 
We construct
a winning strategy as follows (for clarity, we call the initial credit $x$ though
the strategy definition is independent of the value of $x$):

\begin{itemize}
\item[$(1)$] play $\straa_{\gfe}$ until the energy level has increased by 
$\Delta = (n-1) \cdot W$ (i.e., the energy level has reached $x + \Delta$) 
and proceed to $(2)$ with energy level $x' \geq x + \Delta$, 
or play $\straa_{\gfe}$ forever if the energy level never reaches $x + \Delta$;
\item[$(2)$] 
(a) if the current state of the game is not in $\Omega_0$, then play a winning strategy in the subgame
induced by $\Win \setminus \Omega_0$ (which has at most $d-1$ priorities) and such that 
the energy level never drops below $x' - (n-k-1) \cdot W$ where $k = \abs{\Omega_0}$ 
(such a strategy exists by the induction hypothesis);
(b) whenever the game reaches $\Omega_0$, then play a memoryless
strategy to reach a priority-$0$ state (this may decrease the energy level by $k \cdot W$), 
and proceed to $(1)$ with energy level at least $x' - (n-k-1) \cdot W - k \cdot W  = x' - (n-1) \cdot W \geq x$;
\end{itemize}

We show that this strategy is winning in $G$ from every state in $\Win$.
First, we show that the energy level never drops below $x - (n-1) \cdot W \geq 0$ if the initial
credit is $x \geq (n-1) \cdot W$ (and thus in particular never drops below $0$). 
In phase $(1)$, the energy level 
is always at least $x - (n-1) \cdot W \geq 0$ since $\straa_{\gfe}$ is memoryless
and good-for-energy. If the strategy switches to phase $(2)$, then we have
already seen that the energy never drops below $x' - (n-1) \cdot W \geq x$.
Therefore, whenever the strategy switches back to phase $(1)$, the energy level
has not decreased (i.e., it is at least $x$), and the argument can be repeated.
Second, we show that the parity condition is satisfied. We consider three possible
cases: $(i)$ if phases $(1)$ and $(2)$ are played infinitely often, then
priority $0$ is visited infinitely often and the parity condition is satisfied;
$(ii)$ if phase $(1)$ is played finitely often, then eventually phase $(2)$
is played forever, which means that we play a winning strategy in the subgame induced by 
$\Win \setminus \Omega_0$. Therefore, by induction hypothesis the parity condition is satisfied in the
game (since the parity objective is independent of finite prefixes);
$(iii)$ if phase $(2)$ is played finitely often, then eventually phase $(1)$
is played forever, which implies that eventually all visited cycles have weight $0$,
which entails that their least priority is even (by definition of good-for-energy
strategies), hence so is the least priority visited infinitely often.

Now, we analyze the amount of memory needed by this strategy. In this analysis,
we denote by $M(d,n)$ the size of the memory needed by our winning strategy 
in game $G$.
In phase $(1)$,
we need to remember the energy level variation, which is between $ -(n-1) \cdot W$
and $(n-1) \cdot W$, thus can be done with memory of size at most $(2n - 1) \cdot W$.
In phase $(2)$, the subgame strategy has memory size bounded by $M(d-1,n-k)$,
and the attractor strategy is memoryless. Hence, the size of the memory needed
is at most $M(d,n) \leq (2n-1) \cdot W + 1 + M(d-1,n-k)$.

\noindent{{\bf B.}} If the least priority in $G$ is odd (say it is $1$). 
Let $\Win$ be the set of winning states for player~$1$ in $G$ (thus $q_0 \in \Win$),
and let $\Omega_1$ be the player-$2$ attractor of priority-$1$ states in the subgraph of $G$ induced 
by $\Win$. By an argument similar to the proof of 
Lemma~\ref{lem:coBuchi-memoryless}, the set $\Win'$ of states in the subgame induced by $\Win \setminus \Omega_1$
that are winning (for energy parity objective) is nonempty, and player~$1$ is winning in the subgame
induced by $\Win \setminus \Attr_1(\Win')$. 
We construct a winning strategy on $\Attr_1(\Win')$ as follows (for clarity, we call the initial credit $x$ though
the strategy definition is independent of the value of $x$):
\begin{itemize}
\item[$(1)$] play a memoryless strategy to reach $\Win'$ (let $\Delta_1$ be the maximal energy cost), 
and proceed to $(2)$ with energy level $x' = x - \Delta_1$; 
\item[$(2)$] in $\Win'$, play a winning strategy in the subgame
induced by $\Win'$ (which has at most $d-1$ priorities) and such that 
the energy level never drops below $x' - \Delta_2$, where $\Delta_2 = \abs{\Win'} \cdot W$ 
(such a strategy exists by induction hypothesis);
note that $\Delta_1 + \Delta_2 \leq \abs{\Attr_1(\Win')} \cdot W$.
\end{itemize}

We apply the same construction recursively to the subgame induced by $\Win \setminus \Attr_1(\Win')$,
and call the corresponding energy drops $\Delta_3$, $\Delta_4$, etc.

We show that this strategy is winning in $G$ from every state in $\Win$.
First, the sum $\Delta_1 + \Delta_2 + \ldots$ of energy drops is bounded by $\abs{\Win} \cdot W$
and thus initial credit of $(n-1) \cdot W$ is enough.
Second, the parity condition is satisfied since we eventually play a winning strategy in a subgame
where priorities are greater than $1$, and without visiting priority $1$.
The amount of memory needed is the sum of the memory size of the strategies
in the sets $\Win'$ (of size $k_1$, $k_2$, etc.), and of the memoryless attractor strategy (of size $1$), 
hence at most 
$M(d,n) \leq M(d-1, k_1) + \cdots + M(d-1, k_m) + 1$ where $k_1 + \cdots + k_m < n$.

Combining the recurrence relations obtained in  {\bf A} and {\bf B} for $M(d,n)$, 
we verify that $M(d,n) \leq 2 \cdot n \cdot W + n \cdot (d-2) \cdot W \leq n \cdot d \cdot W$
when $d$ is even, and $M(d,n) \leq n \cdot (d-1) \cdot W$ when $d$ is odd.
\qed
\end{proof}

The following theorem summarizes the upper bounds on memory requirement in energy parity games.
Note that player~$1$ may need exponential memory as illustrated in Example~\ref{examp:memory}.

\begin{theorem}[Strategy Complexity] 
For all energy parity games, the following assertions hold:
(1) winning strategies with memory of size $n \cdot d \cdot W$ 
exist for player~1;
(2) memoryless winning strategies exist for player~2.
\end{theorem}

\section{Computational Complexity of Energy Parity Games}

We show that the initial credit problem for energy parity games is in NP $\cap$ coNP. 
The coNP upper bound follows from Lemma~\ref{lem:one-player-epg-P} showing that in player-1 games,
the initial credit problem for energy parity objecitves 
can be solved in polynomial time, and from the fact that memoryless strategies are
sufficient for player~$2$ (Lemma~\ref{lem:player-two-memoryless}). The NP upper bound may be surprising since 
exponential memory may be necessary for player~$1$ to win. 
However, we show that winning strategies with a special structure (that alternate between 
good-for-energy strategies and attractor strategies) can be constructed and 
this entails the NP upper bound. The details are presented after Lemma~\ref{lem:one-player-epg-P}.

\begin{lemma}\label{lem:one-player-epg-P}
The problem of deciding, given a player-$1$ energy parity game~$G$ and an initial
state~$q$, if there exists a finite initial credit such that player~$1$ wins in~$G$
from~$q$ can be solved in polynomial time.
\end{lemma}

\begin{proof}
By the analysis in the proof of Lemma~\ref{lem:one-player-epg},
a polynomial-time algorithm for the initial credit problem in a 
player-$1$ energy parity game~$G$ is as follows: for each even priority $2i$,
consider the decomposition in maximal strongly connected components (scc)
of the restriction of~$G$ to the states with priority at least $2i$. 
The algorithm checks for every scc~$S$ whether $(a)$ $S$ contains a state
with priority $2i$ and a (strictly) positive cycle (using a shortest-path algorithm),
or $(b)$ $S$ contains a cycle of energy level~$0$ through a state with priority $2i$
(again using a shortest-path algorithm). The algorithm returns {\sc Yes} if
for some priority $2i$ and some scc $S$, condition $(a)$ or $(b)$ is satisfied.
The correctness of this algorithm follows from the analysis in the proof
of Lemma~\ref{lem:one-player-epg}, in particular if $(a)$ holds, then a finite
initial credit is sufficient to reach the positive cycle, and repeating this cycles
gives enough energy to visit priority $2i$ in the scc and get back to the cycle,
and if $(b)$ holds, then reaching the energy-$0$ cycle requires finite initial credit, 
and looping through it forever is winning.
Since there are at most $\abs{Q}$ priorities, and at most $\abs{Q}$ scc's,
and since algorithms for scc decomposition and shortest path problem run in 
polynomial time, the desired result follows.
\qed
\end{proof}

\begin{lemma}\label{lem:good-for-credit-NP}
Let $G$ be an energy parity game.
The problem of deciding, given a state $q_0$ and a memoryless strategy $\straa$,
whether $\straa$ is good-for-energy in $q_0$, can be solved in polynomial time.
\end{lemma}

\begin{proof}
Consider the restriction $\widehat{G}$ of the graph $G_{\straa}$ to the 
states reachable from~$q_0$ under strategy~$\straa$. First, for each state $q$ in this graph,
an algorithm for the shortest path problem can be used
to check in polynomial time that every cycle through $q$ has nonnegative
sum of weights. Second, for each state $q$ with odd priority $p(q)$,
the same algorithm checks that every cycle through $q$ in the
restriction of $\widehat{G}$ to the states with priority at least $p(q)$
has (strictly) positive sum of weights.
\qed
\end{proof}

We first establish the NP membership of the initial credit problem for the case
of energy parity games with two priorities. For energy coB\"uchi games, the result
is obtained by showing that memoryless strategies are sufficient, and for 
energy B\"uchi games, the proof gives a good flavor of the argument in the
general case.

\begin{lemma}\label{lem:coBuchi-memoryless}
Memoryless strategies are sufficient for player~$1$ to win energy coB\"uchi games
(i.e., the minimum initial credit for player~$1$ does not change
if player~$1$ is restricted to play memoryless).
\end{lemma}

\begin{proof}
Let $\tuple{G,p,w}$ be an energy coB\"uchi games with $G=\tuple{Q,E}$ (thus $p: Q \to \{1,2\}$). 
Let $\Win \subseteq Q$ be the set of states from which player~$1$ wins in $G$.
Note that the subgraph of $G$ induced by $\Win$ is a game graph.
The proof is by induction on the size of $\Win$. For $\abs{\Win} = 1$,
the result of the lemma is trivial: player~$1$ wins with 
memoryless strategy and initial credit $0$. By induction hypothesis, assume that
for $\abs{\Win} < k$, player~$1$ wins from every state in $\Win$ with memoryless strategy 
and initial credit $(\abs{\Win}-1) \cdot W$. Let $\abs{\Win} = k$.

Let $\Omega_1 \subseteq Q$ be the player-$2$ attractor of priority-$1$ states
(i.e., $\Omega_1$ is the set of states from which player~$2$ can force to reach 
a state with priority $1$).
Consider the subgraph $G'$ of $G$ induced by $Q' = \Win \setminus \Omega_1$. 
We claim that player~$1$ has a (memoryless) winning strategy in the energy game $\tuple{G',w}$
from some state $q'$. We show this by contradiction. Assume that player~$2$ has a (memoryless) spoiling 
strategy $\strab$ on $\Win \setminus \Omega_1$ in the energy game $\tuple{G',w}$, 
and consider the (memoryless) extension of $\strab$ to $\Omega_1$ that enforces to reach priority-$1$ states. 
Let $\rho$ be an outcome of this 
strategy. Either $\rho$ visits $\Omega_1$ (and also priority-$1$ states) infinitely often and 
thus violates the coB\"uchi condition, or $\rho$ eventually stays in $\Win \setminus \Omega_1$
and violates the energy condition. This contradicts that player~$1$ wins in $G$ from $\Win$.
Hence, the set of winning states for player~$1$ in the energy game $\tuple{G',w}$
is nonempty, and player~$1$ is winning with a memoryless strategy (by properties of energy games). 
Note that since all states in $G'$ have priority $2$, this memoryless strategy is also winning for the 
energy coB\"uchi condition. Let $\Win'$ be the player-$1$ attractor of this winning set.
Player~$1$ has a memoryless winning strategy $\straa_e$ from all states in $\Win'$,
and properties of energy games show that an initial credit of $(\abs{\Win'}-1) \cdot W$ is sufficient. 

Now, consider the subgraph $G''$ of $G$ induced by $Q'' = \Win \setminus \Win'$. 
It is easy to see that player~$1$ wins everywhere in the energy coB\"uchi game $\tuple{G'',p,w}$. 
Therefore by induction hypothesis, initial credit $(\abs{Q''}-1) \cdot W$ is sufficient. 
Since $\abs{\Win'} + \abs{Q''} = \abs{\Win}$, player~$1$ can start in any state of $\Win$ with initial credit $(\abs{\Win}-1) \cdot W$
and while the game stays in $G''$, guarantee that the energy level is always at least
$(\abs{\Win'}-1) \cdot W$, so that whenever the game leaves $Q''$, player~$1$ has enough credit to 
use the memoryless winning strategy $\straa_e$ on $\Win'$.
\qed
\end{proof}

\begin{lemma}\label{lem:buchi-cobuchi-np}
The problem of deciding, given a state $q$ in an energy B\"uchi (resp. coB\"uchi) 
game~$G$, if there exists a finite initial credit such that player~$1$ wins in~$G$
from~$q$ is in NP.
\end{lemma}

\begin{proof}
By Lemma~\ref{lem:coBuchi-memoryless}, an NP-algorithm for energy coB\"uchi games $\tuple{G,p,w}$
guesses a memoryless strategy $\straa$ and checks in polynomial time that $\straa$ is winning
for both the energy game $\tuple{G,w}$ and the coB\"uchi game $\tuple{G,p}$. This ensures that
all cycles in $G_\straa$ are positive (for energy) and visit only priority-$2$ states,
and thus $\straa$ is winning in the energy coB\"uchi game.

For energy B\"uchi games, let $\Win$ be the set of winning states for player~$1$ in $\tuple{G,p,w}$,
and let $G_{\Win}$ be the subgraph of $G$ induced by $\Win$.
Clearly there exists a memoryless strategy $\straa_b$ in $G_{\Win}$ that enforces a visit to a priority-$0$ state
from every state in $\Win$, and there exists a memoryless good-for-energy strategy $\straa_{\gfe}$ in $G_{\Win}$ 
(by Lemma~\ref{lem:good-for-energy}). 
We show that the converse holds: if such strategies $\straa_b$ and $\straa_{\gfe}$ exist, then 
player~$1$ wins in the energy B\"uchi game $\tuple{G,p,w}$. 
Let $n=\abs{Q}$ be the number of states.
To prove this,
we give an informal description of a winning strategy for player~$1$ (with initial credit $(n-1) \cdot W$) as follows:
(1) play strategy $\straa_{\gfe}$ as long as the energy level is below $2 \cdot (n-1) \cdot W$;
(2) if the energy level gets higher than $2 \cdot (n-1) \cdot W$, then play $\straa_b$
until a priority-$0$ state is visited (thus $\straa_b$ is played during at most $n-1$ steps), and proceed to step (1)
with energy level at least $(n-1) \cdot W$.

Let $\rho$ be an outcome of this strategy with initial credit 
$(n-1) \cdot W$. 
First, we show that the energy level is nonnegative in every position of $\rho$. 
By definition of good-for-energy strategies, in every cycle of $G_{\straa_{\gfe}}$ the sum of the 
weights is nonnegative. Therefore in the prefixes of $\rho$ corresponding to part $(1)$
of the strategy, the energy level is always nonnegative. Whenever, part $(2)$
of the strategy is played, the energy level is at least $2 \cdot (n-1) \cdot W$ and thus
after (at most) $n-1$ steps of playing $\straa_b$, the energy level is still at 
least $(n-1) \cdot W$, and the argument can be repeated. 
Second, we show that priority-$0$ states are visited infinitely often in $\rho$.
This is obvious if part $(2)$ of the strategy is played infinitely often;
otherwise, from some point in $\rho$, part $(1)$ of the strategy is played
forever which implies that in the cycle decomposition of $\rho$, ultimately
all cycles have sum of weights equal to zero. By definition of good-for-energy strategies, 
every such cycle is even, i.e., visits a priority-$0$ state.

Therefore, an NP-algorithm for energy B\"uchi games guesses the set $\Win \subseteq Q$
and the memoryless strategies $\straa_b$ and $\straa_{\gfe}$ on $\Win$, and checks in polynomial time
using standard graph algorithms
that $\straa_b$ enforces a visit to a priority-$0$ state in $\Win$, that $\straa_{\gfe}$
is good-for-energy (see Lemma~\ref{lem:good-for-credit-NP}), and that $q \in \Win$.
\qed
\end{proof}

\begin{lemma}\label{lem:en-parity-NP}
The problem of deciding, given a state~$q$ in an energy parity game~$G$, if there exists a finite initial credit
such that player~$1$ wins in $G$ from~$q$ is in NP.
\end{lemma}

\begin{proof}
We prove that there exists an NP algorithm that guesses the set of winning states
in~$G$, which entails the lemma.
The result holds for energy parity games
with two priorities by Lemma~\ref{lem:buchi-cobuchi-np}. 
Assume by induction that the result holds for games with less than $d$ priorities, and
let $G$ be an energy parity game with $d$ priorities and $n$ states. 

First, if the least priority in $G$ is even (assume w.l.o.g. that the least priority is $0$), 
an NP algorithm guesses $(i)$ the set $\Win$ of winning states in $G$, and $(ii)$ a memoryless good-for-energy 
strategy $\straa_{\gfe}$ on $\Win$ which must exist by Lemma~\ref{lem:good-for-energy} (this can be done 
in polynomial time by Lemma~\ref{lem:good-for-credit-NP}).
Let $\Omega_0$ be the player-$1$ attractor of priority-$0$ states in the subgraph of $G$ induced 
by $\Win$. By induction, we can check in NP that player~$1$ is winning in the subgraph of $G$ induced 
by $\Win \setminus \Omega_0$ (because this game has less than $d$ priorities). 
This is sufficient to establish that player~$1$ wins in $G$ with initial credit $n \cdot W$, using the following strategy:
(1) play strategy $\straa_{\gfe}$ as long as the energy level is below $2 \cdot n \cdot W$;
(2) while the game is in $\Win \setminus \Omega_0$, we know that player~$1$ can play a winning strategy
that needs initial credit at most $(n-k) \cdot W$ where $k = \abs{\Omega_0}$ and
such that the energy level drops by at most $(n-k-1) \cdot W$ (see also the proof of Lemma~\ref{lem:two-player-epg}),
and therefore 
(3) if the game leaves $\Win \setminus \Omega_0$, then the energy level is at least $2n - (n-k-1) \cdot W = (n+k+1) \cdot W$
which is enough for player~$1$ to survive while enforcing a visit to a 
priority-$0$ state (within at most $\abs{\Omega_0} = k$ steps) and to proceed to step (1)
with energy level at least $n \cdot W$.
Arguments similar to the proof of Lemma~\ref{lem:buchi-cobuchi-np} shows that this strategy is
winning, with initial credit $n \cdot W$. The time complexity of this algorithm 
is $T(n) = p(n) + T(n-1)$ where $p(\cdot)$ is a polynomial (linear) function for the time complexity
of guessing $\Win$ and $\straa_{\gfe}$, checking that $\straa_{\gfe}$ is good-for-energy,
and computing the the player-$1$ attractor of priority-$0$ states $\Omega_0$.
Therefore $T(n) = O(n^2)$.

Second, if the least priority in $G$ is odd (assume w.l.o.g. that the least priority is $1$), 
consider the set $\Win$ of winning states in $G$, and $\Omega_1$ the player~$2$ attractor of priority-$1$
states in the subgame of $G$ induced by $\Win$. By an argument similar to the proof of 
Lemma~\ref{lem:coBuchi-memoryless}, the set $\Win'$ of states in the subgame induced by $\Win \setminus \Omega_1$
that are winning (for energy parity objective) is nonempty, and player~$1$ is winning in the subgame
induced by $\Win \setminus \Attr_1(\Win')$. An NP algorithm guesses the sets $\Win$ and $\Win'$, 
and checks that player~$1$ is winning in $\Win'$ (which can be done in NP, since
$\Win \setminus \Omega_1$ has less than $d$ priorities), and that player~$1$ is winning in $\Win \setminus \Attr_1(\Win')$ which can be done 
in NP, as shown by an induction proof on the number of states in the game since 
the case of games with one state is clearly solvable in NP.
\qed
\end{proof}

\begin{theorem}[Computational Complexity]\label{theo:en-parity-NP-coNP}
 The problem of deciding the existence of a finite initial credit for energy parity games is in NP~$\cap$~coNP.
\end{theorem}

\begin{proof}
By Lemma~\ref{lem:en-parity-NP}, the problem is in NP, and since memoryless strategies
are sufficient for player~$2$ (by Lemma~\ref{lem:player-two-memoryless}), a coNP algorithm
can guess a memoryless spoiling strategy and check in polynomial time that player~$1$
is not winning in the resulting player-$1$ game (by Lemma~\ref{lem:one-player-epg-P}). 
\qed
\end{proof}

\section{Algorithm for Energy Parity Games}

We present an 
algorithm to decide the winner in energy parity games with complexity
exponential in the number of states (as for parity games), and only linear in the largest
weight (as for energy games).
Our algorithm is based on a procedure to construct memoryless good-for-energy strategies.
To obtain a good-for-energy strategy, we modify the weights in the game so that 
every simple cycle with (original) sum of weight $0$ gets a strictly positive weight if it is even, 
and a strictly negative weight if it is odd. 
Winning strategies in the energy game with modified
weights 
correspond to good-for-energy strategies in the original game.

\begin{lemma}\label{lem:good-for-energy-complexity}
The problem of deciding the existence of a memoryless good-for-energy strategy 
in energy parity games can be solved in time $O(\abs{E} \cdot \abs{Q}^{d+1} \cdot W)$.
\end{lemma}

\begin{proof}[of Lemma~\ref{lem:good-for-energy-complexity}]
Given an energy parity game $\tuple{G,p,w}$, we construct a weight function $w'$ such that Player~$1$ 
has a memoryless good-for-energy strategy in $\tuple{G,p,w}$ if and only if Player~$1$ wins in the energy 
game $\tuple{G,w'}$. The maximal weight according to $w'$ becomes $W' = W \cdot \abs{Q}^{d}$, and the complexity result 
then follows from the algorithm of~\cite{CB09,DGR09} which solves energy games in $O(\abs{E} \cdot \abs{Q} \cdot W')$.

Let $0,\dots,d-1$ be the priorities in the energy parity game, and denote by $\EL'$ the energy level function defined according to $w'$.
The function $w'$ is defined by $w'(q,q') = w(q,q') + \Delta(q)$ where
$$\Delta(q) = (-1)^k \cdot \frac{1}{(n+1)^{k+1}}$$
for all $q,q' \in Q$ with $k = p(q)$ and $n = \abs{Q}$.
Note that $\abs{\Delta(q)} < \frac{1}{n^{k+1}}$ for all $q \in Q$ with $k = p(q)$, and in particular $\abs{\Delta(q)} < \frac{1}{n}$ for all $q \in Q$.
Therefore, $n \cdot \abs{\Delta(q)} < 1$ and thus if $\EL(C) < 0$ (i.e., $\EL(C) \leq -1$) for a simple cycle $C$ in~$G$,
then $\EL'(C) < 0$, and if $\EL(C) > 0$, then $\EL'(C) > 0$.  Moreover, if the least priority of a state in $C$ is $k$, then 

$\begin{array}{ll}
\EL'(C) - \EL(C) > \frac{1}{(n+1)^{k+1}} - (n-1) \cdot \frac{1}{(n+1)^{k+2}} >  0 & \text{ if } k \text{ is even, and} {\LARGE \strut}\\
\EL'(C) - \EL(C) < \frac{-1}{(n+1)^{k+1}} + (n-1) \cdot \frac{1}{(n+1)^{k+2}} <  0 & \text{ if } k \text{ is odd.} \\
\end{array}$
\medskip

\noindent So, for simple cycles $C$ with $\EL(C) = 0$, if the least priority in $C$ is even, then $\EL'(C) > 0$,
while if the least priority in $C$ is odd, then $\EL'(C) < 0$. Therefore, a (memoryless) winning strategy in the energy game $\tuple{G,w'}$
can be used as a good-for-energy strategy that avoids odd cycles with sum of weights equal to zero. 
Clearly, the converse also holds, namely if a memoryless strategy is good-for-energy in the energy parity game, 
then it is winning in the energy game. Note that by multiplying the weights in $w'$ by $(n+1)^{d}$, we get integer weights
and the complexity result follows.
\qed
\end{proof}

\begin{algorithm}[t]
\caption{{\sf SolveEnergyParityGame}}
\label{alg:energy-parity-solve}
{
 \AlgData{An energy parity game $\tuple{G,p,w}$ with state space $Q$.}
 \AlgResult{The set of winning states in $\tuple{G,p,w}$ for player~$1$.}
 \Begin{
	\nl \lIf{$Q = \emptyset$}{\KwRet{$\emptyset$}} \;
	\nl Let $k^*$ be the minimal priority in $G$. Assume w.l.o.g. that $k^* \in \{0,1\}$ \;
	\nl Let $G_0$ be the game $G$ \; 
	\nl $i \gets 0$ \;
	\nl \If{$k^* = 0$}
	{
		\nl $A_0 \gets Q \quad$ {\tt /* over-approximation of Player-1 winning states */} \;
		\nl \Repeat{$A_i = A_{i-1}$}
		{
			\nl $A'_i \gets {\sf SolveEnergyGame}(G_{i},w')$ (where $w'$ is defined in Lemma~\ref{lem:good-for-energy-complexity}) \label{alg:Ai}\;
			\nl $X_i \gets \Attr_1(A'_i \cap p^{-1}(0))$ \label{alg:attr1} \;
			\nl Let $G'_{i}$ be the subgraph of $G_{i}$ induced by $A'_{i} \setminus X_i$ \;
			\nl $Z_i \gets (A'_i \setminus X_i) \setminus {\sf SolveEnergyParityGame}(G'_{i}, p, w)$ \;
			\nl $A_{i+1} \gets A'_i \setminus \Attr_2(Z_i)$ \label{alg:attr2}\;
			\nl Let $G_{i+1}$ be the subgraph of $G_{i}$ induced by $A_{i+1}$ \;
			\nl $i \gets i+1$ \;
		}
		\nl \KwRet{$A_i$}\; 
	}
	\nl \If{$k^* = 1$}
	{
		\nl $B_0 \gets Q \quad$ {\tt /* over-approximation of Player-2 winning states */} \;
		\nl \Repeat{$B_i = B_{i-1}$}
		{
			\nl $Y_i \gets \Attr_2(B_i \cap p^{-1}(1))$ \label{alg:attr3}\;
			\nl Let $G_{i+1}$ be the subgraph of $G_{i}$ induced by $B_{i} \setminus Y_i$ \;
			\nl $B_{i+1} \gets B_i \setminus \Attr_1({\sf SolveEnergyParityGame}(G_{i+1}, p, w))$ \label{alg:attr4}\label{alg:CB}\;
			\nl $i \gets i+1$ \;
		}
		\nl \KwRet{$Q \setminus B_i$}\; 
	}
 }
}
\end{algorithm}

We present a recursive fixpoint algorithm for solving energy parity games,
using the result of Lemma~\ref{lem:good-for-energy-complexity}. 
Our algorithm is a generalization of the classical algorithm of McNaughton~\cite{McNaughton93} and Zielonka~\cite{Zielonka98}
for solving parity games. 
The formal description of the algorithm is shown as 
Algorithm~\ref{alg:energy-parity-solve}. 

\smallskip\noindent{\bf Informal description and correctness of Algorithm~\ref{alg:energy-parity-solve}.}
We assume without loss of generality that the least priority in the input game graph is either~0 or~1; 
if not, then we can reduce the priority in every state by~2.
The algorithm considers two cases: (a) when the minimum priority is~0, and (b) when the 
minimum priority is~1.
The details of the two cases are as follows:
\begin{enumerate}
\item[(a)] 
If the least priority in the game is $0$, then we compute
the winning states of Player~$1$ as the limit of a decreasing sequence $A_0,A_1,\dots$ of sets. 
Each iteration removes from $A_i$ some states that are winning for Player~$2$. 
The set $A'_i \subseteq A_i$ contains the states having a good-for-energy strategy (line~\ref{alg:Ai})
which is a necessary condition to win, according to Lemma~\ref{lem:good-for-energy}.
We decompose $A'_i$ into $X_i$ and $A'_i \setminus X_i$, where $X_i$ is the set of states from which
Player~$1$ can force a visit to priority-$0$ states, and $A'_i \setminus X_i$ has less priorities than $A'_i$.
The winning states $Z_i$ in $A'_i \setminus X_i$ for Player~$2$ are also winning in the original game 
(as in $A'_i \setminus X_i$ Player~1 has no edge going out of $A'_i \setminus X_i$).
Therefore we remove $Z_i$ and Player-$2$ attractor to $Z_i$ in $A_{i+1}$. 
The correctness argument for this case is similar to the proof of Lemma~\ref{lem:en-parity-NP}, namely
that when $A_i = A'_i = A_{i-1}$, Player~$1$ wins by playing a winning strategy in $A'_i \setminus X_i$ (which exists by
an inductive argument on the number of recursive calls of the algorithm), and whenever the game enters $X_i$,
then Player~$1$ can survive while forcing a visit to a priority-$0$ state, and then uses a good-for-energy
strategy to recover enough energy to proceed. 

\item[(b)] The second part of the algorithm (when the least priority in the game is $1$) 
computes a decreasing sequence $B_0,B_1,\dots$ 
of sets containing the winning states of Player~$2$. 
The correctness is proven in a symmetric way using the same argument as in the 
second part of the proof of Lemma~\ref{lem:en-parity-NP}.
\end{enumerate}
We obtain the following result, where $d$ is the number of priorities in the game, and $W$ is the 
largest weight.
 
\begin{theorem}[Algorithmic Complexity]\label{theo:en-parity-alg}
The finite initial credit problem for energy parity games 
(i.e., deciding the existence of a finite initial credit) 
can be solved in time $O(\abs{E} \cdot d \cdot \abs{Q}^{d+2} \cdot W)$.
\end{theorem}

\begin{proof}
This problem is solved by Algorithm~\ref{alg:energy-parity-solve}. The key correctness argument 
is given above. The complexity result assumes that good-for-energy strategies can be computed
in time $\GFE(d) = O(\abs{E} \cdot \abs{Q}^{d+1} \cdot W)$ (see Lemma~\ref{lem:good-for-energy-complexity}).

Let $T(d)$ be the complexity of Algorithm~\ref{alg:energy-parity-solve}, parameterized by the number
of priorities in the game. Note that the attractors (lines~\ref{alg:attr1}, \ref{alg:attr2}, \ref{alg:attr3}, \ref{alg:attr4}) 
can be computed in $O(\abs{E})$ which is subsumed by $\GFE(d)$.
Since every recursive call removes at least one state from~$A_i$ (or from~$B_i$), 
there are at most $\abs{Q}$ recursive calls, and since the number of priorities decreases in a recursive call, 
we get 
$$T(d) \leq \abs{Q} \cdot (\GFE(d) + T(d-1))$$ 
and since $\abs{Q} \cdot \GFE(d) = \GFE(d+1)$,
we get $T(d) \leq \GFE(d+1) + \abs{Q} \cdot (T(d-1))$.
Since $d=0$ corresponds to a game with empty state space, we have $T(0) = O(1)$
and it is easy to see that $T(d) \leq d \cdot \GFE(d+1) + \abs{Q}^d$.
The result follows.
\qed
\end{proof}

\noindent{\bf Energy B\"uchi and coB\"uchi games.}
In the special case of energy B\"uchi objectives, since $d$ is constant ($d=2$), the analysis
in the proof of Theorem~\ref{theo:en-parity-alg} gives time complexity $O(\abs{E} \cdot \abs{Q}^{4} \cdot W)$.
In the case of energy coB\"uchi objectives, the smallest priority is~1 and there is only one other 
priority. 
In this case, line~\ref{alg:CB} of Algorithm~\ref{alg:energy-parity-solve} requires
to solve an energy parity game with one priority which can be solved as 
simple energy games in $O(\abs{E} \cdot \abs{Q} \cdot W)$.
Thus in the special case of energy coB\"uchi objectives Algorithm~\ref{alg:energy-parity-solve} 
has $O(\abs{E} \cdot \abs{Q}^{2} \cdot W)$ running time.

\smallskip\noindent{\bf Computing the minimum initial credit.}
Note that if the procedure {\sf SolveEnergyGame} used in Algorithm~\ref{alg:energy-parity-solve} also
computes the minimum initial credit $v(q)$ in each winning state $q$ of the energy game $\tuple{G_{i},w'}$
(and it is the case of the algorithm in~\cite{CB09,DGR09}), then we can also obtain the minimum initial credit 
in the energy parity game $\tuple{G,p,w}$ by rounding $v(q)$ to an integer, either up or 
down. Therefore, computing the minimum initial credit in energy parity games can be done in 
time $O(\abs{E} \cdot d \cdot \abs{Q}^{d+2} \cdot W)$.

\begin{table}[t]
\begin{center}
\scalebox{0.89}{\!
\begin{tabular}{|l|c|c|c|c|}
\hline
             & Player~1 & Player~2  & $\ $ Computational $\ $ & Algorithmic \\
 \multicolumn{1}{|c|}{Objective}  & Strategy & Strategy  & Complexity    & Complexity  \\
\hline
\hline
{\large \strut} Energy coB\"uchi {\large \strut}  &    
Memoryless & {\large \strut} Memoryless {\large \strut} & NP $\cap$ coNP & 
$O(\abs{E} \cdot \abs{Q}^{2} \cdot W)$ \\ 
\hline
{\large \strut} Energy B\"uchi {\large \strut}  &    
Optimal memory: & Memoryless & NP $\cap$ coNP &  
$O(\abs{E} \cdot \abs{Q}^{4} \cdot W)$ \\        
  &  {\large \strut} $2 \! \cdot \! (\abs{Q}-1)\! \cdot \! W +1$ {\large \strut} & & &  \\ 
\hline
{\large \strut} Energy parity   &    
Memory at most: & Memoryless & NP $\cap$ coNP &  
{\large \strut} $O(\abs{E} \cdot d \cdot \abs{Q}^{d+2} \cdot W)$ {\large \strut} \\
& {\large \strut} $\abs{Q}\! \cdot \! d \! \cdot \! W$ {\large \strut} & &  &  \\
\hline
\end{tabular}
}\end{center}
\caption{ Strategy, computational and algorithmic complexity of energy parity games.}\label{tab2}
\end{table}

Our results about the memory requirement of strategies, and the computational and algorithmic complexity
of energy parity games are summarized in Table~\ref{tab2}.

\section{Relationship with Mean-payoff Parity Games}

We show that there is a tight relationship between energy parity games and mean-payoff
parity games. The work in~\cite{ChatterjeeHJ05} shows that optimal\footnote{A strategy is optimal if it is winning
for $\Parity_G(p) \cap \MeanPayoff_G(\nu)$ with the largest possible threshold $\nu$. It is known that the largest 
threshold is rational~\cite{ChatterjeeHJ05}.}
strategies in mean-payoff parity games may require infinite memory, though they can be decomposed into
several memoryless strategies.  
We show that energy parity games are polynomially equivalent to mean-payoff parity games, 
leading to NP~$\cap$~coNP membership of the problem of deciding the winner in mean-payoff parity games,
and leading to
an algorithm for solving such games which is conceptually much simpler than the algorithm of~\cite{ChatterjeeHJ05},
with essentially the same complexity (linear in the largest weight, and exponential in the number of states only).

\begin{theorem}\label{theo:mean-payoff-parity-to-energy-parity}
Let $\tuple{G,p,w}$ be a game, and let $\epsilon = \frac{1}{\abs{Q}+1}$. 
Player~$1$ has a winning strategy in the mean-payoff parity game $\tuple{G,p,w}$ if and only if 
player~$1$ has a winning strategy in the energy parity game $\tuple{G,p,w+\epsilon}$.
\end{theorem}

\begin{proof}
We present the two directions of the proof.
\begin{enumerate}
\item 
Assume that player~$1$ wins from a state~$q_0$ in the mean-payoff parity game
$\tuple{G,p,w}$. Then, for all $\varepsilon > 0 $ there exists a finite-memory
winning strategy $\straa$ in $G$ with threshold $-\varepsilon$ from 
$q_0$~\cite{ChatterjeeHJ05}. 
Consider a finite-memory winning strategy $\straa$ for 
$\epsilon=\frac{1}{\abs{Q}+1}$. 
We show that $\straa$ is winning in the energy parity game 
$\tuple{G,p,w+\epsilon}$ from $q_0$.

Consider the graph $G_\straa$. By definition of $\straa$, 
the average of the weights (according to $w$) in all cycles of 
$G_\straa$ reachable from $q_0$ is at least $-\epsilon$, and the least priority in every such cycle is even. 
Therefore, in every outcome of $\straa$ from $q_0$, the parity condition is 
satisfied, and the sum of weights (according to $w+\epsilon$) is nonnegative, 
hence $\straa$ is winning in the energy parity game $\tuple{G,p,w+\epsilon}$ 
from $q_0$, with initial credit $\abs{G_\straa}\cdot W$, where $\abs{G_\straa}$
denotes the number of states in $G_\straa$.

\item 
Assume that player~$1$ wins from a state~$q_0$ in the energy parity game
$\tuple{G,p,w+\epsilon}$. Then, there exists a finite-memory strategy 
$\straa$ in $G$ from $q_0$ that ensures in $G_{\straa}$ that all cycles
reachable from $q_0$ have least priority even, and nonnegative sum of weights (according to
$w+\epsilon$), i.e., the average of the weights (according to $w$) is 
at least $-\epsilon$. Therefore, the value of strategy $\straa$ 
in the mean-payoff parity game $\tuple{G,p,w}$ from $q_0$ is at least $-\epsilon$.

Now, the results of~\cite{ChatterjeeHJ05} show that the optimal value 
that player~$1$ can ensure in a mean-payoff parity game is a rational 
number of the form $\frac{e}{d}$ such that $1 \leq d \leq \abs{Q}$ 
and $\abs{e} \leq d \cdot W$. 
It follows that if the value for mean-payoff parity games is greater than 
$\frac{1}{\abs{Q}}$, then the value is at least~0.
Since $\epsilon < \frac{1}{\abs{Q}}$, it follows that
there must exist a strategy for player~$1$ in $\tuple{G,p,w}$ from $q_0$ 
with value at least $0$, hence player~$1$ is winning in the mean-payoff parity 
game $\tuple{G,p,w}$ from $q_0$.
\end{enumerate}
The result follows.
\qed
\end{proof}

\begin{corollary} Given a mean-payoff parity game, whether player~1 has a 
winning strategy from a state $q_0$ can be decided in NP~$\cap$~coNP.
\end{corollary}

\begin{corollary}
The problem of deciding the winner in mean-payoff parity games can be solved in time
$O(\abs{E} \cdot d \cdot \abs{Q}^{d+2} \cdot W \cdot (\abs{Q}+1))$.
\end{corollary}

\smallskip\noindent{\bf Acknowledgements.} 
We thank Thomas A. Henzinger and Barbara Jobstmann for inspiring discussions,
and Patricia Bouyer, Nicolas Markey, J\"org Olschewski, and Michael Ummels
for helpful comments on a preliminary draft.

\bibliographystyle{plain}
\bibliography{biblio}



\end{document}